\documentclass{llncs}
\usepackage{amsmath}
\usepackage{times}
\usepackage{indentfirst}
\usepackage{graphicx}
\usepackage{amssymb}

\newtheorem{faulse assertion}{faulse assertion}
\newtheorem{counter-example}{counter-example}

\begin{document}

\title{Unique expansion matroids and union minimal matroids}         

\author{Hua Yao, William Zhu\thanks{Corresponding author.
E-mail: williamfengzhu@gmail.com(William Zhu)} }
\institute{Lab of Granular Computing,\\
Zhangzhou Normal University, Zhangzhou, China}



\date{\today}          
\maketitle

\begin{abstract}
The expansion axiom of matroids requires only the existence of some kind of independent sets, not the uniqueness of them. This causes that the base families of some matroids can be reduced while the unions of the base families of these matroids remain unchanged. In this paper, we define unique expansion matroids in which the expansion axiom has some extent uniqueness; we define union minimal matroids in which the base families have some extent minimality. Some properties of them and the relationship between them are studied.
First, we propose the concepts of secondary base and forming base family. Secondly, we propose the concept of unique expansion matroid, and prove that a matroid is a unique expansion matroid if and only if its forming base family is a partition. Thirdly, we propose the concept of union minimal matroid, and prove that unique expansion matroids are union minimal matroids. Finally, we extend the concept of unique expansion matroid to unique exchange matroid and prove that both unique expansion matroids and their dual matroids are unique exchange matroids.
\end{abstract}

\section{Introduction}
Matroids were introduced by Whitney~\cite{Whitney1935Ontheabstract} in 1935 to try to
capture abstractly the essence of dependence. Since then, it has been recognized that matroids arise naturally in combinatorial optimization. Matroids have been applied to diverse fields such as algorithm design~\cite{Edmonds71Matroids}, combinatorial optimization~\cite{Lawler01Combinatorialoptimization}. Recently, matroids have been combined with rough sets~\cite{HuangZhu12Geometriclattice,LiuZhuZhang12Relationshipbetween,TangSheZhu12matroidal,WangZhu11Matroidal,ZhuWang11Rough}.

There are several ways to define a matroid, and independent set axiom is one of them. It presents three properties of independent set, and the third one is called expansion axiom. It indicates that if there exist two independent sets whose cardinalities are not equal and the independent set whose cardinality is smaller is not a subset of the other, there exists some other independent set of which the independent set whose cardinality is smaller is a proper subset. Generally, such independent sets are more than one, for the expansion axiom does not require that such independent set is unique. This causes that the base families of some matroids can be reduced while the unions of the base families of these matroids remain unchanged.

In this paper, we define unique expansion matroids in which the expansion axiom has some extent uniqueness; we define union minimal matroids in which the base families have some extent minimality. Some properties of them and the relationship between them are studied. First, we propose the concepts of secondary base and forming base family, and study some properties of forming base families in detail. Secondly, we propose the concept of unique expansion matroid. The expression of the base families of this type of matroids is presented. We prove that a matroid is a unique expansion matroid if and only if its forming base family is a partition. Thirdly, we propose the concept of unique partition matroid. We prove that a matroid is a unique expansion matroid if and only if it is a unique partition matroid. Fourthly, we propose the concepts of union minimal matroid and intersection minimal matroid. That union minimal matroid and intersection minimal matroid are dual is proved. We prove that unique expansion matroids are union minimal matroids. Finally, we extend the concept of unique expansion matroid to unique exchange matroid and prove that both unique expansion matroids and their dual matroids are unique exchange matroids.

The remainder of this paper is organized as follows. In Section~\ref{S:Preliminaries}, we review the relevant concepts. In Section~\ref{Secondary base unique expansion matroid}, we give the concepts of secondary base, forming base family, unique expansion matroid and unique partition matroid. Then we study the properties of them and the relationship between them. In Section~\ref{Union minimal matroid}, we give the concepts of union minimal matroid and intersection minimal matroid. Then we prove that unique expansion matroids are union minimal matroids. In Section~\ref{Unique exchange matroid}, we extend the concept of unique expansion matroid to unique exchange matroid and prove that both unique expansion matroids and their dual matroids are unique exchange matroids. Section~\ref{S:Conclusions} presents conclusions.

\section{Preliminaries}
\label{S:Preliminaries}

For a better understanding to this paper, in this section, the concepts of covering and partition and some basic knowledge of matroids are introduced. In this paper, we denote $\cup_{X\in S}X$ by $\cup S$, where $S$ is a set family.

\begin{definition}(Covering)
\label{definition67}
Let $E$ be a universe of discourse and $\mathbf{C}$ be a family of subsets of $E$.
If $\emptyset\notin\mathbf{C}$ and $\cup \mathbf{C}=E$, $\mathbf{C}$ is called a covering of $E$.
Every element of $\mathbf{C}$ is called a covering block.
\end{definition}

In the following discussion, unless stated to the contrary, the universe of discourse $E$ is considered to be
finite and nonempty. If it is demanded that any two blocks of a covering have no common elements, we obtain the concept of partition.

\begin{definition}(Partition)
\label{definition68}
Let $E$ be a universe of discourse and $\mathbf{P}$ be a family of subsets of $E$.
$\mathbf{P}$ is called a partition of $E$ if the following conditions hold: $(1)$ $\emptyset\notin\mathbf{P}$; $(2)$
$\cup \mathbf{P}=E$; $(3)$ for any $K,L\in\mathbf{P}$, $K\cap L=\emptyset$.
Every element of $\mathbf{P}$ is called a partition block.
\end{definition}

It is obvious a partition of $E$ is certainly a covering of $E$. So the concept of covering is an extension of the concept of partition.
Some operational symbols in set theory will be used in this paper. We introduce the definitions of them as follows.

\begin{definition}(~\cite{Lai01Matroid})
\label{definition1}
Let $E$ be a set and $\mathcal{A}$ a family of subsets of $E$. Three operators are defined as follows:\\
$Low(\mathcal{A})=\{X\subseteq E|\exists A(A\in\mathcal{A}\wedge X\subseteq A)\}$, \\
$Max(\mathcal{A})=\{X\in\mathcal{A}|\forall Y(Y\in\mathcal{A}\wedge X\subseteq Y\rightarrow X=Y)\}$,\\ $Com(\mathcal{A})=\{X\subseteq E|E-X\in\mathcal{A}\}$.
\end{definition}

There are several ways to define matroids, and the following is one of them.

\begin{definition}(Matroid~\cite{Lai01Matroid})
\label{definition2}
A matroid $M$ is an ordered pair $(E,\mathcal{I})$, where $E$ is a finite set, $\mathcal{I}$ is a collection of subsets of $E$ and $\mathcal{I}$ satisfies the following three properties:\\
(I1) $\emptyset\in\mathcal{I}$;\\
(I2) if $I\in\mathcal{I}$ and $I^{\prime}\subseteq I$, $I^{\prime}\in\mathcal{I}$;\\
(I3) if $I_{1},I_{2}\in\mathcal{I}$ and $|I_{1}|<|I_{2}|$, there exists $e\in I_{2}-I_{1}$ such that $I_{1}\cup\{e\}\in\mathcal{I}$.
\end{definition}

Every element of $\mathcal{I}$ is called an independent set of matroid $M$. Matroid $M$ is usually denoted as $M(E,\mathcal{I})$. Sometimes, $\mathcal{I}$ in $M(E,\mathcal{I})$ is denoted as $\mathcal{I}(M)$; $E$ in $M(E,\mathcal{I})$ is denoted as $E(M)$.

\begin{definition}(Base~\cite{Lai01Matroid})
\label{definition3}
Let $M(E,\mathcal{I})$ be a matroid. $Max(\mathcal{I})$ is denoted as $\mathcal{B}(M)$ and called the base family of $M(E,\mathcal{I})$. Any $B\in\mathcal{B}(M)$ is called a base of $M(E,\mathcal{I})$.
\end{definition}

The following proposition indicates that all the bases have the same cardinality.

\begin{proposition}(~\cite{Lai01Matroid})
\label{proposition100}
Let $M$ be a matroid and $B_{1},B_{2}\in\mathcal{B}(M)$. Then $|B_{1}|=|B_{2}|$.
\end{proposition}

By the above proposition, we give the following definition.

\begin{definition}
\label{definition101}
Let $M$ be a matroid and $B\in\mathcal{B}(M)$. $|B|$ is denoted as $r(M)$ and called the rank of $M$.
\end{definition}

For the commutativity of bases, we have the following two theorems.

\begin{theorem}(~\cite{Lai01Matroid})
\label{theorem4}
Let $E$ be a finite and nonempty set and $\mathcal{B}\subseteq2^{E}$. Then there exists a matroid $M(E,\mathcal{I})$ such that $\mathcal{B}=\mathcal{B}(M)$ iff $\mathcal{B}$ satisfies the following two properties:\\
(B1) $\mathcal{B}\neq\emptyset$;\\
(B2) if $B_{1},B_{2}\in\mathcal{B}$ and $x\in B_{1}-B_{2}$, there exists $y\in B_{2}-B_{1}$ such that $(B_{1}-\{x\})\cup\{y\}\in\mathcal{B}$.
\end{theorem}

\begin{theorem}(~\cite{Lai01Matroid})
\label{theorem123}
Let $M(E,\mathcal{I})$ be a matroid, $B_{1},B_{2}\in\mathcal{B}(M)$ and $x\in B_{1}-B_{2}$. Then there exists $y\in B_{2}-B_{1}$ such that $(B_{2}-\{y\})\cup\{x\}\in\mathcal{B}$.
\end{theorem}

For $Com(\mathcal{B}(M))$, we have the following theorem.

\begin{theorem}(~\cite{Lai01Matroid})
\label{theorem7}
Let $M(E,\mathcal{I})$ be a matroid. Then $(E,Low(Com(\mathcal{B}(M))))$ is a matroid.
\end{theorem}

According to the above theorem, we introduce the concept of dual matroid.

\begin{definition}(Dual matroid~\cite{Lai01Matroid})
\label{definition8}
Let $M(E,\mathcal{I})$ be a matroid.
$(E,Low(Com(\mathcal{B}(M))))$ is denoted as $M^{*}$ and called the dual matroid of $M$.
\end{definition}

Sometimes, we denote $\mathcal{I}(M^{*})$ and $\mathcal{B}(M^{*})$ as $\mathcal{I}^{*}(M)$ and $\mathcal{B}^{*}(M)$, respectively.

\section{Unique expansion matroid}
\label{Secondary base unique expansion matroid}

In this section, we propose the concepts of secondary base, forming base family, unique expansion matroid and unique partition matroid. Then we study their properties and the relationship between them.

\subsection{Secondary base and forming base family}

Secondary bases are a type of independent sets. We give its definition as follows.

\begin{definition}(Secondary base)
\label{definition43}
Let $M$ be a matroid and $r(M)>0$. $\{A\in\mathcal{I}(M)||A|=r(M)-1\}$ is denoted as $s(M)$ and called the  secondary base family of $M$. Any $A\in s(M)$ is called a secondary base of $M$.
\end{definition}

We propose an operator on matroids.

\begin{definition}
\label{definition44}
Let $M(E,\mathcal{I})$ be a matroid. An operator $K_{M}:2^{E}\rightarrow2^{E}$ is defined by:
for any $X\subseteq E$, $K_{M}(X)=\{a\in E|r(X\cup\{a\})=r(X)+1\}$.
\end{definition}

It is obvious $X\cap K_{M}(X)=\emptyset$. By the above definition, we give the concept of forming base family.

\begin{definition}(Forming base family)
\label{definition77}
Let $M$ be a matroid and $r(M)>0$. The forming base family of $M$ is defined by: $F(M)=\{K_{M}(X)|X\in s(M)\}$.
\end{definition}

Given a base of a matroid $M$, we obtain a subset of $F(M)$.

\begin{definition}
\label{definition340}
Let $M$ be a matroid, $r(M)>0$ and $B\in\mathcal{B}(M)$. The forming base family of $M$ with respect of $B$ is defined by: $F_{M}(B)=\{K_{M}(X)|X\in s(M)\wedge X\subseteq B\}$.
\end{definition}

It is obvious $F_{M}(B)\subseteq F(M)$. The forming base family with respect of a base has the following property.

\begin{proposition}
\label{proposition341}
$|F_{M}(B)|=r(M)$.
\end{proposition}

For proving the above proposition, we firstly prove the following simple lemma.

\begin{lemma}
\label{lemma66}
If $r(M)=1$, $F_{M}(B)=\{\cup\mathcal{B}(M)\}$.
\end{lemma}

\begin{proof}
By $r(M)=1$, we have that
$F_{M}(B)=\{K_{M}(X)|X\in s(M)\wedge X\subseteq B\}=\{K_{M}(\emptyset)\}=\{\cup\mathcal{B}(M)\}$.
$\Box$
\end{proof}

The proof of Proposition~\ref{proposition341} is presented as follows.

\begin{proof}
If $r(M)=1$, by Lemma~\ref{lemma66}, this proposition follows. If $r(M)>1$, we let $B=\{b_{1},b_{2},\cdots,b_{t}\}$, where $t>1$. For any $1\leq i<j\leq t$, we have that $b_{i}\in K_{M}(B-\{b_{i}\})-K_{M}(B-\{b_{j}\})$. Thus $K_{M}(B-\{b_{i}\})\neq K_{M}(B-\{b_{j}\})$. Then $|F_{M}(B)|=|\{K_{M}(B-\{b_{1}\}),K_{M}(B-\{b_{2}\}),\cdots,K_{M}(B-\{b_{t}\})\}|=r(M)$. $\Box$

\end{proof}

By Proposition~\ref{proposition341}, we have the following corollary.

\begin{corollary}
\label{corollary423}
$|F(M)|\geq r(M)$.
\end{corollary}

\begin{proof}

It follows form $F_{M}(B)\subseteq F(M)$ and Proposition~\ref{proposition341}. $\Box$

\end{proof}

The cardinality of the forming base family of a matroid can be greater than the rank of the matroid, as well as equal to the rank of the matroid. To illustrate this, let us see an example.

\begin{example}
\label{example300}
Let $E=\{1,2,3\}$, $\mathcal{B}_{1}=\{\{1,2\},\{1,3\}\}$ and $M_{1}=(E,Low(\mathcal{B}_{1}))$. Then $M_{1}$ is a matroid and $Q_{M_{1}}=\{\{1\},\{2,3\}\}$. So $|Q_{M_{1}}|=2=r(M_{1})$. Let $\mathcal{B}_{2}=\{\{1,2\},\{1,3
\},\{2,
3\}\}$ and $M_{2}=(E,Low(\mathcal{B}_{2}))$. Then $M_{2}$ is a matroid and $Q_{M_{2}}=\{\{1,2\},\{1,3\},\{2,3\}\}$. So $|Q_{M_{2}}|=3>2=r(M_{2})$.
\end{example}

Based on this, it is a natural issue that under what conditions the cardinality of the forming base family of a matroid is equal to the rank of the matroid. With the discussion getting further, we will give a necessary and sufficient condition for this issue. Now we continue to discuss the properties of forming base families.

\begin{proposition}
\label{proposition46}
$\cup F(M)=\cup\mathcal{B}(M)$.
\end{proposition}

\begin{proof}

For any $g\in\cup F(M)$, we know that there exists some $A\in s(M)$ such that $g\in K_{M}(A)$. Hence $\{g\}\cup A\in\mathcal{B}(M)$. Thus $g\in\cup\mathcal{B}(M)$, hence $\cup F(M)\subseteq\cup\mathcal{B}(M)$. For any $h\in\cup\mathcal{B}(M)$, we know that there exists some $B\in\mathcal{B}(M)$ such that $h\in B$. Let $A=B-\{h\}$. It is obvious $A\in s(M)$ and $h\in K_{M}(A)$. Thus $h\in\cup F(M)$, therefore $\cup\mathcal{B}(M)\subseteq\cup F(M)$. So $\cup F(M)=\cup\mathcal{B}(M)$. $\Box$

\end{proof}

To illustrate the above proposition, let us see an example.

\begin{example}
\label{example47}
Let $E=\{1,2,3\}$, $\mathcal{B}=\{\{1,2\},\{1,3\},\{2,3\}\}$ and $M=(E,Low(\mathcal{B}))$. Then $M$ is a matroid and $F(M)=\{\{1,2\},\{1,3\},\{2,3\}\}$. $\cup F(M)$ is a covering on $\cup\mathcal{B}(M)$, not a partition.
\end{example}

Based on this, it is a natural issue that under what conditions a forming base family is a partition. To address this issue, we need to propose the concept of unique expansion matroid. On the other hand, there is a more general result than that in Proposition~\ref{proposition46}.

\begin{proposition}
\label{proposition124}
For any $B\in\mathcal{B}(M)$, $\cup F_{M}(B)=\cup\mathcal{B}(M)$.
\end{proposition}

For proving the above proposition, we firstly prove the following lemma, which indicates that for any element $b\in B$, there exists only one element $K\in F_{M}(B)$ such that $b\in K$.

\begin{lemma}
\label{lemmaE}
Let $b\in B\in\mathcal{B}(M)$. Then $|\{K\in F_{M}(B)|b\in K\}|=1$.
\end{lemma}

\begin{proof}
If $r(M)=1$, by Lemma~\ref{lemma66}, this proposition follows. If $r(M)>1$, we let  $B=\{b,b_{2},\cdots,b_{t}\}$, where $t>1$. It is obvious  $F_{M}(B)=\{K_{M}(B-\{b\}),K_{M}(B-\{b_{2}\}),\cdots,K_{M}(B-\{b_{t}\})\}$ and $b\in K_{M}(B-\{b\})$. For any $2\leq j\leq t$, we have that $b\notin K_{M}(B-\{b_{j}\})$. Thus $|\{K\in F_{M}(B)|b\in K\}|=1$. $\Box$

\end{proof}

The proof of Proposition~\ref{proposition124} is presented as follows.

\begin{proof}
By $F_{M}(B)\subseteq F(M)$ and Proposition~\ref{proposition46}, we have that $\cup F_{M}(B)\subseteq\cup\mathcal{B}(M)$. By Lemma~\ref{lemmaE}, we know that $B\subseteq\cup F_{M}(B)$. For any  $d\in\cup\mathcal{B}(M)-B$, we know that there exists some $D\in\mathcal{B}(M)$ such that $d\in D$. Thus $d\in D-B$. By Theorem~\ref{theorem123}, we know that there exists some $b\in B-D$ such that $(B-\{b\})\cup\{d\}\in\mathcal{B}(M)$. Hence $d\in K_{M}(B-\{b\})$. Since $K_{M}(B-\{b\})\in F_{M}(B)$,  $d\in\cup F_{M}(B)$. Hence $\cup\mathcal{B}(M)\subseteq\cup F_{M}(B)$. Therefore $\cup F_{M}(B)=\cup\mathcal{B}(M)$. $\Box$

\end{proof}

\subsection{Unique expansion matroid}

By the definition of matroids, for any secondary base and any base, there exists at least one element of the base which does not belong to the secondary base such that the union of the secondary base and this element is a base. Of course, there may exist more than one element of the base which satisfies the conditions. Below is an example.

\begin{example}
\label{example48}
Let $E=\{1,2,3\}$, $\mathcal{B}=\{\{1,2\},\{1,3\},\{2,3\}\}$ and $M=(E,Low(\mathcal{B}))$. Then $M$ is a matroid. Let $I_{2}=\{2,3\}$ and $I_{1}=\{1\}$. It is obvious $2\in I_{2}-I_{1}$, $3\in I_{2}-I_{1}$,  $(I_{1}\cup\{2\})\in\mathcal{I}(M)$ and $(I_{1}\cup\{3\})\in\mathcal{I}(M)$.
\end{example}

For any secondary base and any base, if there exists just one element of the base which satisfies the above conditions, we obtain a special type of matroids.

\begin{definition}(Unique expansion matroid)
\label{definition49}
Let $M$ be a matroid. For any $B\in\mathcal{B}(M)$ and any $A\in s(M)$, if by $e_{1}\in B$, $e_{2}\in B$, $A\cup\{e_{1}\}\in\mathcal{B}(M)$ and $A\cup\{e_{2}\}\in\mathcal{B}(M)$, we obtain that $e_{1}=e_{2}$, $M$ is called a unique expansion matroid.
\end{definition}

To illustrate this concept, let us see an example.

\begin{example}
\label{example73}
Let $E=\{1,2,3\}$, $\mathcal{B}=\{\{1,2\},\{1,3\}\}$ and $M=(E,Low(\mathcal{B}))$. Then $M$ is a unique expansion  matroid.
\end{example}

With the concept of unique expansion matroid, we can answer the question that under what conditions a forming base family is a partition.

\begin{theorem}
\label{theorem50}
$F(M)$ is a partition on $\cup\mathcal{B}(M)$ iff $M$ is a unique expansion matroid.
\end{theorem}

\begin{proof}

($\Rightarrow$): We use the proof by contradiction. Suppose that $M$ is not a unique expansion matroid. Then there exists some $A\in s(M)$ and some $B\in\mathcal{B}(M)$ such that there exists some $e_{1}\in B$ and some $e_{2}\in B$, where $e_{1}\neq e_{2}$, such that $A\cup\{e_{1}\}\in\mathcal{B}(M)$ and $A\cup\{e_{2}\}\in\mathcal{B}(M)$. It is obvious $\{e_{1},e_{2}\}\in\mathcal{I}(M)$. By (I3) of Definition~\ref{definition2}, we know that there exists some $D\subset A$, where $|D|=|A|-1$, such that $D\cup\{e_{1},e_{2}\}\in\mathcal{B}(M)$. Let $A-D=\{a\}$, $A_{1}=D\cup\{e_{1}\}$ and $A_{2}=D\cup\{e_{2}\}$. Then $a\in K_{M}(A_{1})$ and $a\in K_{M}(A_{2})$. Hence $a\in K_{M}(A_{1})\cap K_{M}(A_{2})$. Since $e_{2}\in K_{M}(A_{1})-K_{M}(A_{2})$, $K_{M}(A_{1})\neq K_{M}(A_{2})$. Therefore $F(M)$ is not a partition.

($\Leftarrow$): By Proposition~\ref{proposition46}, we know that $F(M)$ is a covering on $\cup\mathcal{B}(M)$. So we need to prove only that for any $K_{M}(A_{1}),K_{M}(A_{2})\in F(M)$, if $K_{M}(A_{1})\cap K_{M}(A_{2})\neq\emptyset$, $K_{M}(A_{1})=K_{M}(A_{2})$. If ${A_{1}}={A_{2}}$, the conclusion is obviously true. Below we suppose that ${A_{1}}\neq{A_{2}}$. Let $a\in K_{M}(A_{1})\cap K_{M}(A_{2})$. Then $\{a\}\cup A_{1}\in\mathcal{B}(M)$ and $\{a\}\cup A_{2}\in\mathcal{B}(M)$. For any $b\in K_{M}(A_{1})$, we will prove that $b\in K_{M}(A_{2})$. If $b=a$, the conclusion is obviously true. Below we suppose $b\neq a$. We claim that $b\notin A_{2}$. Otherwise, suppose $b\in A_{2}$. Then $b\in A_{2}\cup\{a\}$. Since $a\in A_{2}\cup\{a\}$, $A_{1}\cup\{b\}\in\mathcal{B}(M)$ and $A_{1}\cup\{a\}\in\mathcal{B}(M)$, by $M$ is a unique expansion matroid, we have that $b=a$. It is
contradictory. Thus $b\notin A_{2}\cup\{a\}$, therefore $b\in(A_{1}\cup\{b\})-(A_{2}\cup\{a\})$. By Theorem~\ref{theorem123}, we know that there exists some $g\in(A_{2}\cup\{a\})-(A_{1}\cup\{b\})$ such that $((A_{2}\cup\{a\})-\{g\})\cup\{b\}\in\mathcal{B}(M)$. We claim that $g=a$. Otherwise, suppose $g\neq a$. We have that $a\in((A_{2}\cup\{a\})-\{g\})\cup\{b\}$. Since $b\in((A_{2}\cup\{a\})-\{g\})\cup\{b\}$, $A_{1}\cup\{b\}\in\mathcal{B}(M)$ and $A_{1}\cup\{a\}\in\mathcal{B}(M)$, by $M$ is a unique expansion matroid, we have that $b=a$. It is
contradictory. Thus $g=a$. Hence $((A_{2}\cup\{a\})-\{g\})\cup\{b\}=A_{2}\cup\{b\}$, thus
$b\in K_{M}(A_{2})$. Therefore $K_{M}(A_{1})\subseteq K_{M}(A_{2})$. Similarly, $K_{M}(A_{2})\subseteq K_{M}(A_{1})$. So $K_{M}(A_{1})=K_{M}(A_{2})$. Then $F(M)$ is a partition on $\cup\mathcal{B}(M)$. $\Box$

\end{proof}

The following proposition presents a property of the bases of unique expansion matroids.

\begin{proposition}
\label{propositionH}
Let $M$ be a unique expansion matroid. For any $B\in\mathcal{B}(M)$ and any $D\in F(M)$, $|B\cap D|=1$.
\end{proposition}

\begin{proof}

Let $D=K_{M}(A)$, where $A\in s(M)$. By (I3) of Definition~\ref{definition2}, there exists some $g\in B-A$ such that $A\cup\{g\}\in\mathcal{B}(M)$. Hence $g\in K_{M}(A)$, thus $g\in B\cap K_{M}(A)$. For any $h\in B\cap K_{M}(A)$, we have that $h\in B$ and $A\cup\{h\}\in\mathcal{B}(M)$. Since $M$ is a unique expansion matroid, $g=h$. Then $B\cap K_{M}(A)=\{g\}$. Hence $|B\cap D|=1$. $\Box$

\end{proof}

Now we can answer the question that under what conditions the cardinality of the forming base family of a matroid is equal to the rank of the matroid.

\begin{theorem}
\label{theorem126}
$|F(M)|=r(M)$ iff $M$ is a unique expansion matroid.
\end{theorem}

\begin{proof}

($\Rightarrow$): By Theorem~\ref{theorem50}, we need to prove only that $F(M)$ is a partition on $\cup\mathcal{B}(M)$. By Proposition~\ref{proposition46}, we know that $F(M)$ is a covering on $\cup\mathcal{B}(M)$. We use the proof by contradiction. Suppose $F(M)$ is not a partition on $\cup\mathcal{B}(M)$. Then there exist some $K_{p},K_{q}\in F(M)$, where $K_{p}\neq K_{q}$, such that $K_{p}\cap K_{q}\neq\emptyset$. Suppose $b\in K_{p}\cap K_{q}$. It is obvious there exists some $B\in\mathcal{B}(M)$ such that $b\in B$. By Lemma~\ref{lemmaE} and Proposition~\ref{proposition341}, we know that $|F(M)|\geq|\{K_{p},K_{q}\}\cup F_{M}(B)|\geq r(M)+1$. It is contradictory.

($\Leftarrow$): Let $B\in\mathcal{B}(M)$ and $F(M)=\{D_{1},D_{2},\cdots,D_{t}\}$. By Proposition~\ref{propositionH}, we know that for any $1\leq i\leq t$, it follows that $|B\cap D_{i}|=1$. Let $\{d_{i}\}=B\cap D_{i}$.
By Theorem~\ref{theorem50}, we know that $F(M)$ is a partition on $\cup\mathcal{B}(M)$. Hence for any $1\leq i<j\leq t$, $d_{i}\neq d_{j}$. Thus $|F(M)|\leq r(M)$. Suppose $|F(M)|<r(M)$. Since $B\subseteq\cup\mathcal{B}(M)=\cup F(M)$, there exists some $F\in F(M)$ such that $|B\cap F|>1$. It is contradictory. Hence $|F(M)|=r(M)$. $\Box$

\end{proof}

The following proposition gives a sufficient condition for a subset of $E(M)$ to be a base of a unique expansion  matroid $M$.

\begin{proposition}
\label{propositionJ}
Let $M$ be a unique expansion matroid, $B\subseteq\cup\mathcal{B}(M)$ and for any $D\in F(M)$, it follows that $|B\cap D|=1$. Then $B\in\mathcal{B}(M)$.
\end{proposition}

\begin{proof}

By Theorem~\ref{theorem126}, we know that $|F(M)|=r(M)$.
Let $r(M)=t$ and $F(M)=\{P_{1},P_{2},\cdots,P_{t}\}$. Again let $B\cap P_{i}=\{b_{i}\}$, where $1\leq i\leq t$. Hence $B=\{b_{1},b_{2},\cdots,b_{t}\\
\}$. Let $1\leq s\leq t-1$. We claim that if $\{b_{1},b_{2},\cdots,b_{s}\}\in\mathcal{I}(M)$, $\{b_{1},b_{2},\cdots,b_{s},b_{s+1}\}\in\mathcal{I}(M)$. It is obvious there exists some $B_{s}\in\mathcal{B}(M)$ such that $\{b_{1},b_{2},\cdots,b_{s}\}\subseteq B_{s}$. By Proposition~\ref{propositionH}, we know that $|B_{s}\cap P_{s+1}|=1$.
Let $B_{s}\cap P_{s+1}=\{d\}$. Since $d\in K_{M}(B_{s}-\{d\})\in F(M)$, $d\in P_{s+1}$ and that $F(M)$ is a partition, $K_{M}(B_{s}-\{d\})=P_{s+1}$. By $b_{s+1}\in P_{s+1}$, we have that $b_{s+1}\in K_{M}(B_{s}-\{d\})$. Hence $(B_{s}-\{d\})\cup\{b_{s+1}\}\in\mathcal{B}(M)$. Since $\{b_{1},b_{2},\cdots,b_{s},b_{s+1}\}\subseteq(B_{s}-\{d\})\cup\{b_{s+1}\}$, $\{b_{1},b_{2},\cdots,b_{s},\\
b_{s+1}\}\in\mathcal{I}(M)$.
By $b_{1}\in B\subseteq\cup\mathcal{B}(M)$, we have that  $\{b_{1}\}\in\mathcal{I}(M)$. Thus $\{b_{1},b_{2}\}\in\mathcal{I}(M)$. We know that this procedure can carry on until $\{b_{1},b_{2},\cdots,b_{t}\}\in\mathcal{I}(M)$. By $r(M)=t$, we have that  $\{b_{1},b_{2},\cdots,b_{t}\}\in\mathcal{B}(M)$. $\Box$

\end{proof}

By Proposition~\ref{propositionH} and Proposition~\ref{propositionJ}, we obtain a necessary and sufficient condition for a subset of $E(M)$ to be a base of a unique expansion matroid $M$.

\begin{proposition}
\label{proposition51}
Let $M$ be a unique expansion matroid. Then $B\in\mathcal{B}(M)$ iff
$B\subseteq\cup\mathcal{B}(M)$ and for any $D\in F(M)$, it follows that $|B\cap D|=1$.
\end{proposition}

By the above proposition, we give a relationship between the base family and the forming base family of a unique expansion matroid. It indicates that any base of a unique expansion matroid can be obtained by selecting one and only one element from every block of $F(M)$.

\begin{proposition}
\label{proposition125}
Let $M$ be a unique expansion matroid and $F(M)=\{K_{1},K_{2},
\cdots,K_{t}\}$.
Then $\mathcal{B}(M)=\{\{b_{1},b_{2},
\cdots,b_{t}\}|b_{i}\in K_{i},1\leq i\leq t\}$.
\end{proposition}

\begin{proof}

By Theorem~\ref{theorem50}, we know that $F(M)$ is a partition on $\cup\mathcal{B}(M)$. Let $H=\{\{b_{1},b_{2},\\
\cdots,b_{t}\}|b_{i}\in K_{i},1\leq i\leq t\}$. For any $B\in\mathcal{B}(M)$ and any $K_{i}\in F(M)$, by Proposition~\ref{proposition51}, we know that $|B\cap K_{i}|=1$. Let $\{b_{i}\}=B\cap K_{i}$. We have that $\{b_{1},b_{2},\cdots,b_{t}\}\subseteq B$. Suppose $B-\{b_{1},b_{2},\cdots,b_{t}\}\neq\emptyset$. Then $|B|>t$. By Proposition~\ref{proposition46}, we have that $B\subseteq\cup F(M)$. Then there exists some $K_{j}\in F(M)$ such that $|B\cap K_{j}|>1$. It is contradictory. Hence $B=\{b_{1},b_{2},\cdots,b_{t}\}$, thus $B\in H$, therefore $\mathcal{B}(M)\subseteq H$. For any $\{b_{1},b_{2},\cdots,b_{t}\}\in H$, by Proposition~\ref{proposition46}, we have that $\{b_{1},b_{2},\cdots,b_{t}\}\subseteq\cup\mathcal{B}(M)$. For any $K_{i}\in F(M)$, it is obvious $|\{b_{1},b_{2},\cdots,b_{t}\}\cap K_{i}|=1$. By Proposition~\ref{proposition51}, we know that $\{b_{1},b_{2},\cdots,b_{t}\}\in\mathcal{B}(M)$. Thus $H\subseteq\mathcal{B}(M)$. Hence $\mathcal{B}(M)=H$. Then $\mathcal{B}(M)=\{\{b_{1},b_{2},\cdots,b_{t}\}|b_{i}\in K_{i},1\leq i\leq t\}$. $\Box$

\end{proof}

Unique expansion matroids are defined by certain properties, not by
specific structures. We want to know whether or not there exist some existing matroids which are unique expansion matroids. In the following subsection, we will answer this question.

\subsection{Unique partition matroid}

After the concept of matroid was proposed, many types of matroids were constructed. Partition matroids introduced in~\cite{LiuChen94Matroid,LiuZhuZhang12Relationshipbetween} are one type of them. In this subsection, we firstly introduce the concept of partition matroid and give a necessary and sufficient condition for a subset of $E(M)$ to be a base of a partition matroid $M$. Then we focus on studying a special type of partition matroids. Finally, it is shown that this special type of partition matroids and unique expansion matroids are the same.

\begin{proposition}(~\cite{LiuChen94Matroid,LiuZhuZhang12Relationshipbetween})
\label{proposition302}
Let $E$ be a finite set and $P=\{P_{1},P_{2},\cdots,P_{m}\}$ be a partition on $E$. Let $k_{1},\cdots,k_{m}$ be a group of nonnegative integers, which satisfy $k_{i}\leq|P_{i}|$. Let $\mathcal{I}(P;k_{1},\cdots,k_{m})=\{X\subseteq E||X\cap P_{i}|\leq k_{i},1\leq i\leq m\}$. Then $(E,\mathcal{I}(P;k_{1},\cdots,k_{m}))$ is a matroid.
\end{proposition}

By the above proposition, we introduce the following definition.

\begin{definition}(Partition matroid)
\label{definition69}
Matroid $(E,\mathcal{I}(P;k_{1},\cdots,k_{m}))$ is denoted as $M(P;\\
k_{1},\cdots,
k_{m})$ and called a partition matroid.
\end{definition}

The following proposition gives a necessary and sufficient condition for a subset of $E(M)$ to be a base of a partition matroid $M$.

\begin{proposition}
\label{proposition303}
$B\in\mathcal{B}(M(P;k_{1},\cdots,k_{m}))$ iff $B\subseteq \cup P$ and for any $P_{i}\in P$, it follows that $|B\cap P_{i}|=k_{i}$.
\end{proposition}

\begin{proof}

($\Rightarrow$): It is obvious $B\subseteq \cup P$. Suppose there exists some $P_{i}\in P$ such that $|B\cap P_{i}|\neq k_{i}$. Then $|B\cap P_{i}|> k_{i}$ or $|B\cap P_{i}|< k_{i}$. If $|B\cap P_{i}|> k_{i}$, then $B\notin\mathcal{I}(M(P;k_{1},\cdots,k_{m}))$. It is contradictory. If $|B\cap P_{i}|< k_{i}$, then for any $b\in P_{i}-B$, $|(B\cup\{b\})\cap P_{j}|\leq k_{j}$, where $1\leq j\leq m$. Thus $B\cup\{b\}\in\mathcal{I}(M(P;k_{1},\cdots,k_{m}))$. It is contradictory.

($\Leftarrow$): By Proposition~\ref{proposition302}, we know that $B\in\mathcal{I}(M(P;k_{1},\cdots,k_{m}))$. For any $d\in E-B$, without loss of generality, we suppose $d\in P_{l}\in P$. Since $|B\cap P_{l}|=k_{l}$, then $|(B\cup\{d\})\cap P_{l}|=k_{l}+1$. Hence $B\cup\{d\}\notin\mathcal{I}(M(P;k_{1},\cdots,k_{m}))$. Therefore $B\in\mathcal{B}(M(P;k_{1},\cdots,k_{m}))$. $\Box$

\end{proof}

Now we consider a special type of partition matroids.

\begin{proposition}
\label{proposition304}
Let $E$ be a finite set, $P$ a partition on $\cup P$ and $\cup P\subseteq E$. Let $\mathcal{I}_{P}=\{X\subseteq \cup P||X\cap D|\leq1,D\in P\}$. Then $(E,\mathcal{I}_{P})$ is a matroid.
\end{proposition}

\begin{proof}

Let $P=\{D_{1},D_{2},\cdots,D_{t}\}$, $D_{t+1}=E-\cup P$ and $Q=P\cup\{D_{t+1}\}$. Let $k_{i}=1$, where $1\leq i\leq t$ and $k_{t+1}=0$. Then $Q$ is a partition on $E$ and $\{X\subseteq\cup P||X\cap D_{i}|\leq1,1\leq i\leq t\}=\{X\subseteq E||X\cap D_{i}|\leq k_{i},1\leq i\leq t+1\}$. Hence $(E,\mathcal{I}_{P})=M(Q;k_{1},\cdots,k_{t+1})$. $\Box$
\end{proof}

\begin{definition}(Unique partition matroid)
\label{definition70}
Matroid $(E,\mathcal{I}_{P})$ is denoted as $M_{E}(P)$ and called a unique partition matroid.
\end{definition}

By Proposition~\ref{proposition303}, we have the following proposition.

\begin{proposition}
\label{proposition305}
$B\in\mathcal{B}(M_{E}(P))$ iff $B\subseteq\cup P$ and for any $K\in P$, it follows that $|B\cap K|=1$.
\end{proposition}

\begin{proof}
It is obvious $B\in\mathcal{B}(M_{E}(P))$ iff $B\in\mathcal{B}(M_{\cup P}(P))$. Again by Proposition~\ref{proposition303}, this proposition has been proved. $\Box$
\end{proof}

By contrast, we have a necessary and sufficient condition for a subset of $E$ to be a base of the dual matroid of a unique partition matroid $M_{E}(P)$.

\begin{proposition}
\label{proposition339}
$B\in\mathcal{B}^{*}(M_{E}(P))$ iff $E-\cup P\subseteq B\subseteq E$ and for any $K\in P$, it follows that $|K-B|=1$.
\end{proposition}

\begin{proof}

($\Rightarrow$): By $B\in\mathcal{B}^{*}(M_{E}(P))$, we have that $E-B\subseteq\mathcal{B}(M_{E}(P))$. Thus $E-B\subseteq\cup P$, hence $E-\cup P\subseteq B$. It is obvious $B\subseteq E$. Therefore $E-\cup P\subseteq B\subseteq E$. For any $K\in P$, by Proposition~\ref{proposition305}, we have that $|K-B|=|(E\cap K)-(B\cap K)|=|(E-B)\cap K|=1$.

($\Leftarrow$): By $E-\cup P\subseteq B\subseteq E$, we have that $E-B\subseteq\cup P$. For any $K\in P$, we have that $|(E-B)\cap K|=|(E\cap K)-(B\cap K)|=|K-B|=1$. By Proposition~\ref{proposition305}, we have that $E-B\in\mathcal{B}(M_{E}(P))$. Then $B\in\mathcal{B}^{*}(M_{E}(P))$. $\Box$

\end{proof}

Now we give the expression of the base family of a unique partition matroid. It indicates that any base of a unique partition matroid can be obtained by selecting one and only one element from every block of the given partition.

\begin{proposition}
\label{proposition306}
Let $P=\{D_{1},D_{2},\cdots,D_{t}\}$ be a partition on $\cup P$. Then $\mathcal{B}(M_{E}(P))=\{\{b_{1},b_{2},\cdots,b_{t}\}|b_{i}\in D_{i},1\leq i\leq t\}$.
\end{proposition}

\begin{proof}

Let $W=\{\{b_{1},b_{2},\cdots,b_{t}\}|b_{i}\in D_{i},1\leq i\leq t\}$. For any $B\in\mathcal{B}(M_{E}(P))$ and any $D_{i}\in P$, where $1\leq i\leq t$, by Proposition~\ref{proposition305}, we have that $|B\cap D_{i}|=1$. Suppose $B\cap D_{i}=\{b_{i}\}$. Then $\{b_{1},b_{2},\cdots,b_{t}\}\subseteq B$. Suppose $B-\{b_{1},b_{2},\cdots,b_{t}\}\neq\emptyset$ and $a\in B-\{b_{1},b_{2},\cdots,b_{t}\}$. By Proposition~\ref{proposition305}, we know that $B\subseteq\cup P$. Then there exists some $D_{j}\in P$ such that $a\in D_{j}$. Thus $a\in B\cap D_{j}$, hence $a=b_{j}$. It is contradictory. Thus $B=\{b_{1},b_{2},\cdots,b_{t}\}$, therefore $\mathcal{B}(M_{E}(P))\subseteq W$. For any $\{c_{1},c_{2},\cdots,c_{t}\}\in W$, it is obvious $\{c_{1},c_{2},\cdots,c_{t}\}\subseteq \cup P$ and for any $D_{i}\in P$, $|\{c_{1},c_{2},\cdots,c_{t}\}\cap D_{i}|=1$. By Proposition~\ref{proposition305}, we have that $\{c_{1},c_{2},\cdots,c_{t}\}\in\mathcal{B}(M_{E}(P))$. Hence $W\subseteq\mathcal{B}(M_{E}(P))$, thus $\mathcal{B}(M_{E}(P))=W=\{\{b_{1},b_{2},\cdots,b_{t}\}|b_{i}\in D_{i},1\leq i\leq t\}$. $\Box$

\end{proof}

By the above proposition, we obtain the following corollary.

\begin{corollary}
\label{corollary336}
$\cup\mathcal{B}(M_{E}(P))=\cup P$.
\end{corollary}

By Proposition~\ref{proposition306}, we can prove that the forming base family of a unique partition matroid is just the partition which induces the matroid.

\begin{theorem}
\label{theorem321}
$Q(M_{E}(P))=P$.
\end{theorem}

\begin{proof}

Let $P=\{D_{1},D_{2},\cdots,D_{t}\}$. For any $V\in Q(M_{E}(P))$, we know that there exists some $A\in F_{M_{E}(P)}$ such that $V=K_{M}(A)$. By Proposition~\ref{proposition306}, we know that there exists some $j$, where $1\leq j\leq t$, such that $A=\{b_{i_{1}},b_{i_{2}},\cdots,b_{i_{t-1}}\}$, where $\{i_{1},i_{2},\cdots,i_{t-1}\}=\{1,2,\cdots,t\}-\{j\}$ and for any $i_{s}\in\{i_{1},i_{2},\cdots,i_{t-1}\}$, it follows that $b_{i_{s}}\in D_{i_{s}}$. Again by Proposition~\ref{proposition306}, we know that $V=K_{M}(A)=D_{j}$. Hence $V\in P$, therefore $Q(M_{E}(P))\subseteq P$. For any $D_{j}\in P$, where $1\leq j\leq t$, we know that $D_{j}=K_{M_{E}(P)}(\{b_{i_{1}},b_{i_{2}},\cdots,b_{i_{t-1}}\})$, where $\{i_{1},i_{2},\cdots,i_{t-1}\}=\{1,2,\cdots,t\}-\{j\}$ and for any $i_{s}\in\{i_{1},i_{2},\cdots,i_{t-1}\}$, it follows that $b_{i_{s}}\in D_{i_{s}}$. Hence $D_{j}\in Q(M_{E}(P))$, therefore $P\subseteq Q(M_{E}(P))$. Then $Q(M_{E}(P))=P$. $\Box$

\end{proof}

Now we can answer the question that whether or not there exist some existing matroids which are unique expansion matroids.

\begin{theorem}
\label{theorem52}
$M$ is a unique expansion matroid iff $M$ is a unique partition matroid.
\end{theorem}

\begin{proof}

($\Rightarrow$): By Theorem~\ref{theorem50}, we know that $F(M)$ is a partition on $\cup\mathcal{B}(M)$. For any $B\in\mathcal{B}(M)$, by Proposition~\ref{proposition125}, we know that $B\subseteq\cup F(M)$ and for any $D\in F(M)$, it follows that $|B\cap D|=1$. Then for any $X\in\mathcal{I}(M)$, we have that $X\subseteq\cup F(M)$ and for any $D\in F(M)$, it follows that $|X\cap D|\leq1$. Thus $\mathcal{I}(M)=\{X\subseteq\cup F(M)||X\cap D|\leq1,D\in F(M)\}$, therefore $M$ is a unique partition matroid.

($\Leftarrow$): It follows from Theorem~\ref{theorem321} and Theorem~\ref{theorem50}. $\Box$

\end{proof}

\section{Union minimal matroid}
\label{Union minimal matroid}

In this section, we propose the concepts of union minimal matroid and intersection minimal matroid, which are collectively called minimal matroid. We will prove that these two types of minimal matroids are dual and unique expansion matroids are union minimal matroids. For a given matroid, if we remove some bases from the base family and keep the union of the base family of the matroid unchanged, can the remainder be the base family of another matroid? The following example indicates sometimes it can, and sometimes it can not.

\begin{example}
\label{example333}
Let $E=\{1,2,3\}$, $\mathcal{B}_{1}=\{\{1,2\},\{1,3\},\{2,3\}\}$ and $M_{1}=(E,Low
(\mathcal{B}_{1}))$. Then $M_{1}$ is a matroid. Remove $\{2,3\}$ from $\{\{1,2\},\{1,3\},\{2,3\}\}$, we obtain $\{\{1,2\},\{\\
1,3
\}\}$. Let $\mathcal{B}_{2}=\{\{1,2\},\{1,3\}\}$ and $M_{2}=(E,Low(\mathcal{B}_{2}))$. Then $M_{2}$ is still a matroid and $\cup\mathcal{B}_{1}=\cup\mathcal{B}_{2}$. But remove any base from $\mathcal{B}(M_{2})$, the remainder is not the base family of any matroid $M$ which satisfies $\cup\mathcal{B}(M)=\cup\mathcal{B}(M_{2})$.
\end{example}

Based on this, we propose the following concept.

\begin{definition}(Union minimal matroid)
\label{definition37}
Let $M(E,\mathcal{I})$ be a matroid. For any matroid $M_{1}(E,\mathcal{I}_{1})$, if by $\cup\mathcal{B}(M_{1})=\cup\mathcal{B}(M)$ and $\mathcal{B}(M_{1})\subseteq\mathcal{B}(M)$, we obtain that  $\mathcal{B}(M_{1})=\mathcal{B}(M)$, $M$ is called a union minimal matroid.
\end{definition}

The following proposition presents a property of the isomorphism of union minimal matroids.

\begin{proposition}
\label{proposition42}
There exist two matroids $M_{1}(E,\mathcal{I}_{1})$ and $M_{2}(E,\mathcal{I}_{2})$, which satisfy that both $M_{1}$ and $M_{2}$ are union minimal matroids, $\cup\mathcal{B}(M_{1})=\cup\mathcal{B}(M_{2})$ and $r(M_{1})=r(M_{2})$, but $M_{1}$ and $M_{2}$ are not isomorphic.
\end{proposition}

\begin{proof}

We need only to give an example which satisfies the hypothesis given in this proposition. Let $E=\{1,2,3,4,5\}$, $\mathcal{I}_{1}=Low\{\{1,2\},\{1,3\},\{1,4\}\}$, $\mathcal{I}_{2}=Low\{\{
1,3\\
\},\{1,4\},\{2,3\},\{2,4\}\}$, $M_{1}=(E,\mathcal{I}_{1})$ and
$M_{2}=(E,\mathcal{I}_{2})$. Then both $M_{1}$ and $M_{2}$ are union minimal matroids. It is obvious  $\cup\mathcal{B}(M_{1})=\{1,2,3,4\}=\cup\mathcal{B}(M_{2})$ and $r(M_{1})=2=r(M_{2})$. But by $\cap\mathcal{B}(M_{1})=\{1\}$ and $\cap\mathcal{B}(M_{2})=\emptyset$, we know that $M_{1}$ and $M_{2}$ are not isomorphic. $\Box$

\end{proof}

By contrast, we propose the following concept.

\begin{definition}(Intersection minimal matroid)
\label{definition111}
Let $M(E,\mathcal{I})$ be a matroid. For any matroid $M_{1}(E,\mathcal{I}_{1})$, if by  $\cap\mathcal{B}(M_{1})=\cap\mathcal{B}(M)$ and $\mathcal{B}(M_{1})\subseteq\mathcal{B}(M)$, we obtain that  $\mathcal{B}(M_{1})=\mathcal{B}(M)$, $M$ is called an intersection minimal matroid.
\end{definition}

The following proposition presents a property of the isomorphism of intersection minimal matroids.

\begin{proposition}
\label{proposition42}
There exist two matroids $M_{1}(E,\mathcal{I}_{1})$ and $M_{2}(E,\mathcal{I}_{2})$, which satisfy that both $M_{1}$ and $M_{2}$ are intersection union minimal matroids, $\cap\mathcal{B}(M_{1})=\cap\mathcal{B}(M_{2})$ and $r(M_{1})=r(M_{2})$, but $M_{1}$ and $M_{2}$ are not isomorphic.
\end{proposition}

\begin{proof}

We need only to give an example which satisfies the hypothesis given in this proposition. Let $E=\{1,2,3,4,5\}$, $\mathcal{I}_{1}=Low\{\{2,3\},\{2,4\},\{3,4\}\}$, $\mathcal{I}_{2}=Low\{\{1,3\\
\},\{1,4\},\{2,3\},\{2,4\}\}$,
$M_{1}=(E,\mathcal{I}_{1})$ and
$M_{2}=(E,\mathcal{I}_{2})$. Then both $M_{1}$ and $M_{2}$ are intersection minimal matroids. It is obvious $\cap\mathcal{B}(M_{1})=\cap\mathcal{B}(M_{2})=\emptyset$ and $r(M_{1})=2=r(M_{2})$. But by $\cup\mathcal{B}(M_{1})=\{2,3,4\}$ and $\cup\mathcal{B}(M_{2})=\{1,2,3,4\}$, we know that $M_{1}$ and $M_{2}$ are not isomorphic. $\Box$

\end{proof}

The following theorem presents the relationship between union minimal matroids and intersection minimal matroids.

\begin{theorem}
\label{theorem334}
$M(E,\mathcal{I})$ is a union minimal matroid iff $M^{*}$ is an intersection minimal matroid.
\end{theorem}

\begin{proof}
For any $B\in\mathcal{B}(M)$, denote $E-B$ as $B^{*}$. Since $B\in\mathcal{B}(M)\Leftrightarrow B^{*}\in\mathcal{B}(M^{*})$, then
$\cup\mathcal{B}(M_{1})=\cup\mathcal{B}(M)\\
\Leftrightarrow\cup_{B_{i}\in\mathcal{B}(M_{1})}B_{i}=\cup_{B\in\mathcal{B}(M)}B\\
\Leftrightarrow E-\cup_{B_{i}\in\mathcal{B}(M_{1})}B_{i}=E-\cup_{B\in\mathcal{B}(M)}B\\
\Leftrightarrow\cap_{B_{i}\in\mathcal{B}(M_{1})}(E-B_{i})=\cap_{B\in\mathcal{B}(M)}(E-B)\\
\Leftrightarrow\cap_{B_{i}\in\mathcal{B}(M_{1})}(B_{i}^{*})=\cap_{B\in\mathcal{B}(M)}B^{*}\\
\Leftrightarrow\cap_{B_{i}^{*}\in\mathcal{B}(M_{1}^{*})}(B_{i}^{*})=\cap_{B^{*}\in\mathcal{B}(M^{*})}B^{*}\\
\Leftrightarrow\cap\mathcal{B}(M_{1}^{*})=\cap\mathcal{B}(M^{*})$ and\\
$\mathcal{B}(M_{1})\subseteq\mathcal{B}(M)\\
\Leftrightarrow\forall B_{i}(B_{i}\in\mathcal{B}(M_{1})\rightarrow B_{i}\in\mathcal{B}(M))\\
\Leftrightarrow\forall B_{i}(B_{i}^{*}\in\mathcal{B}(M_{1}^{*})\rightarrow B_{i}^{*}\in\mathcal{B}(M^{*}))\\
\Leftrightarrow\forall B_{i}^{*}(B_{i}^{*}\in\mathcal{B}(M_{1}^{*})\rightarrow B_{i}^{*}\in\mathcal{B}(M^{*}))\\
\Leftrightarrow\mathcal{B}(M_{1}^{*})\subseteq\mathcal{B}(M^{*})$ and\\
$\mathcal{B}(M_{1})=\mathcal{B}(M)\Leftrightarrow\mathcal{B}(M_{1}^{*})=\mathcal{B}(M^{*})$. Therefore\\
$M$ is a union minimal matroid\\
$\Leftrightarrow \forall M_{1}(E,\mathcal{I}_{1})((\cup\mathcal{B}(M_{1})=\cup\mathcal{B}(M)\wedge \mathcal{B}(M_{1})\subseteq\mathcal{B}(M))\rightarrow(\mathcal{B}(M_{1})=\mathcal{B}(M)))\\
\Leftrightarrow\forall M_{1}(E,\mathcal{I}_{1})((\cap\mathcal{B}(M_{1}^{*})=\cap\mathcal{B}(M^{*})\wedge \mathcal{B}(M_{1}^{*})\subseteq\mathcal{B}(M)^{*})\rightarrow(\mathcal{B}(M_{1}^{*})=\mathcal{B}(M^{*})))\\
\Leftrightarrow\forall M_{1}^{*}((\cap\mathcal{B}(M_{1}^{*})=\cap\mathcal{B}(M^{*})\wedge \mathcal{B}(M_{1}^{*})\subseteq\mathcal{B}(M)^{*})\rightarrow(\mathcal{B}(M_{1}^{*})=\mathcal{B}(M^{*})))\\
\Leftrightarrow M^{*}$ is an intersection minimal matroid. $\Box$

\end{proof}

In Propositions~\ref{proposition51} and~\ref{proposition125} or Propositions~\ref{proposition305} and~\ref{proposition306}, we see that if a matroid is a unique expansion matroid or a unique partition matroid, there exists such an expression of the base family and such a property for the base family. The following theorem indicates that the expression and the property are equivalent.

\begin{theorem}
\label{theorem33}
Let $M(E,\mathcal{I})$ be a matroid and $P=\{K_{1},K_{2},\cdots,K_{t}\}$ is a partition on $\cup\mathcal{B}(M)$. Then for any $B\in\mathcal{B}(M)$ and any $K_{i}\in P$, $|B\cap K_{i}|=1$ iff  $\mathcal{B}(M)=\{\{b_{1},b_{2},\cdots,b_{t}\}|b_{i}\in K_{i},1\leq i\leq t\}$.
\end{theorem}

\begin{proof}

($\Rightarrow$): We will firstly prove that $r(M)=t$. For any $B\in\mathcal{B}(M)$, we know that $B\subseteq\cup P$. If $r(M)<t$, there exists some $K_{i}\in P$ such that $|B\cap K_{i}|=0$. This is a contradiction
to the hypothesis. If $r(M)>t$, there exists some $K_{j}\in P$ such that $|B\cap K_{j}|\geq2$. This is a contradiction
to the hypothesis. Thus $r(M)=t$. Let $H=\{\{b_{1},b_{2},\cdots,b_{t}\}|b_{i}\in K_{i},1\leq i\leq t\}$. For any $B\in\mathcal{B}(M)$ and any $K_{i}\in P$, we let $\{b_{i}\}=B\cap K_{i}$. Thus $\{b_{1},b_{2},\cdots,b_{t}\}\subseteq B$. By $r(M)=t$, we have that $B=\{b_{1},b_{2},\cdots,b_{t}\}$. Thus $B\in H$. Therefore $\mathcal{B}(M)\subseteq H$.
Below we will prove that $H\subseteq\mathcal{B}(M)$. For any $\{a_{1},a_{2},\cdots,a_{t}\}\in H$, where $a_{i}\in K_{i}$ and $1\leq i\leq t$, since $a_{1}\in\cup\mathcal{B}(M)$, there exists some $B\in\mathcal{B}(M)$ such that $a_{1}\in B$. Since $\mathcal{B}(M)\subseteq H$, $B=\{a_{1},d_{2},\cdots,d_{t}\}$, where $d_{i}\in K_{i}$ and $2\leq i\leq t$. Let $|\{d_{i}\in B|d_{i}\neq a_{i},2\leq i\leq t\}|=v$. We use induction on $v$. If $v=1$, without
loss of generality, suppose $\{d_{i}\in B|d_{i}\neq a_{i},2\leq i\leq t\}=\{d_{2}\}$. Since $a_{2}\in\cup\mathcal{B}(M)$, there exists some $D\in\mathcal{B}(M)$ such that $a_{2}\in D$. Hence $a_{2}\in D-B$. By Theorem~\ref{theorem123}, we know that there exists some $g\in B-D$ such that $(B-\{g\})\cup\{a_{2}\}\in\mathcal{B}(M)$. We claim that $g=d_{2}$. Otherwise we have that $\{a_{2},d_{2}\}\subseteq(B-\{g\})\cup\{a_{2}\}$. Hence $|((B-\{g\})\cup\{a_{2}\})\cap K_{2}|\geq2$. This is a contradiction to the hypothesis. Hence $(B-\{g\})\cup\{a_{2}\}=\{a_{1},a_{2},\cdots,a_{t}\}$. Therefore $\{a_{1},a_{2},\cdots,a_{t}\}\in\mathcal{B}(M)$. Assume that $\{a_{1},a_{2},\cdots,a_{t}\}\in\mathcal{B}(M)$ if $|\{d_{i}\in B|d_{i}\neq a_{i},2\leq i\leq t\}|=v-1$. Now we suppose $|\{d_{i}\in B|d_{i}\neq a_{i},2\leq i\leq t\}|=v$. Without
loss of generality, suppose $d_{2}\in\{d_{i}\in B|d_{i}\neq a_{i},2\leq i\leq t\}$. We know that there exists some $D\in\mathcal{B}(M)$ such that $a_{2}\in D$. Hence $a_{2}\in D-B$. By Theorem~\ref{theorem123}, we know that there exists some $g\in B-D$ such that $(B-\{g\})\cup\{a_{2}\}\in\mathcal{B}(M)$. We claim that $g=d_{2}$. Otherwise we have that $\{a_{2},d_{2}\}\subseteq(B-\{g\})\cup\{a_{2}\}$. Hence $|((B-\{g\})\cup\{a_{2}\})\cap K_{2}|\geq2$. This is a contradiction to the hypothesis. Let $B_{1}=(B-\{g\})\cup\{a_{2}\}=\{a_{1},a_{2},\cdots,d_{t}\}$. Then $|\{d_{i}\in B_{1}|d_{i}\neq a_{i},2\leq i\leq t\}|=v-1$. By the assumption of induction, we have that $\{a_{1},a_{2},\cdots,a_{t}\}\in\mathcal{B}(M)$. Hence $H\subseteq\mathcal{B}(M)$. So $\mathcal{B}(M)=H$, therefore $\mathcal{B}(M)=\{\{b_{1},b_{2},\cdots,b_{t}\}|b_{i}\in K_{i},1\leq i\leq t\}$.

($\Leftarrow$): It follows obviously. $\Box$

\end{proof}

In order to indicate the above theorem from another perspective, we propose a new concept.

\begin{definition}(Combination number of a partition)
\label{definition335}
Let $P$ be a partition on $\cup P$. $\prod_{K\in P}\\
|K|$ is denoted as $Co(P)$ and called the combination number of $P$.
\end{definition}

By Theorem~\ref{theorem33} and the above definition, we have the following corollary.

\begin{corollary}
\label{corollary109}
Let $M$ be a matroid. If there exists a partition $P$ on $\cup\mathcal{B}(M)$ such that for any $B\in\mathcal{B}(M)$ and any $K\in P$, that $|B\cap K|=1$ follows, $|\mathcal{B}(M)|=Co(P)$.
\end{corollary}

Suppose that there exist two partitions which both satisfy the conditions given in Theorem~\ref{theorem33}. The following proposition indicates that they are equal.

\begin{proposition}
\label{proposition103}
Let $M$ be a matroid and $P$ and $Q$ be two partitions on $\cup\mathcal{B}(M)$. Let $P$ and $Q$ satisfy that for any $B\in\mathcal{B}(M)$, any $K\in P$ and any $L\in Q$, it follows that $|B\cap K|=1$ and $|B\cap L|=1$. Then $P=Q$.
\end{proposition}

\begin{proof}

We use the proof by contradiction. Suppose $P\neq Q$. Then there exists some $a\in \cup\mathcal{B}(M)$, $K\in P$ and $L\in Q$ such that $a\in K\cap L$ and $K\neq L$. Then $K-L\neq\emptyset$ or $L-K\neq\emptyset$. Without
loss of generality, suppose $K-L\neq\emptyset$. Let $b\in K-L$. By $a\in L$, $b\notin L$ and  Theorem~\ref{theorem33}, we know that there exists some $B\in\mathcal{B}(M)$ such that $\{a,b\}\subseteq B$. Thus $|B\cap K|\geq2$. This is contradictory. $\Box$

\end{proof}

The following theorem gives a relationship between unique expansion matroids and union minimal matroids.

\begin{theorem}
\label{theorem552}
If $M$ is a unique expansion matroid, $M$ is a union minimal matroid.
\end{theorem}

\begin{proof}
For any $B\in\mathcal{B}(M)$ and any $K\in F(M)$, by Proposition~\ref{proposition51}, we know that $|B\cap K|=1$. Let $M_{1}(E,\mathcal{I}_{1})$ be a matroid which satisfies $\mathcal{B}(M_{1})\subseteq\mathcal{B}(M)$ and $\cup\mathcal{B}(M_{1})=\cup\mathcal{B}(M)$. By Proposition~\ref{proposition46}, we know that $\cup F(M)=\cup\mathcal{B}(M)=\cup\mathcal{B}(M_{1})$. Since $\mathcal{B}(M_{1})\subseteq\mathcal{B}(M)$, for any $D\in\mathcal{B}(M_{1})$, we have that $D\in\mathcal{B}(M)$. Hence for any $K\in F(M)$, we have that $|D\cap K|=1$. By Corollary~\ref{corollary109}, we know that $|\mathcal{B}(M_{1})|=Co(F(M))=|\mathcal{B}(M)|$. Since $\mathcal{B}(M)$ is a finite set, $\mathcal{B}(M_{1})=\mathcal{B}(M)$. Therefore $M$ is a union minimal matroids. $\Box$

\end{proof}

Since unique partition matroids and unique expansion matroids are the same, unique partition matroids are union minimal matroids.

\section{Unique exchange matroid}
\label{Unique exchange matroid}

In this section, we extend the concept of unique expansion matroid to unique exchange matroid and prove that both unique expansion matroids and their dual matroids are unique exchange matroids. Let us see an example.

\begin{example}
\label{example338}
Let $E=\{1,2,3,4,5\}$, $\mathcal{B}=\{\{1,2,3\},\{1,2,4\},\{1,3,4\},\{1,2,5\},\{1,
4,\\
5\}\}$ and $M=(E,Low(\mathcal{B}))$. Then $M$ is a matroid. Let $B_{1}=\{1,2,3\}$, $B_{2}=\{1,4,5\}$, $x=3$, $y_{1}=4$ and $y_{2}=5$. It is obvious $x\in B_{1}-B_{2}$, $y_{1}\in B_{2}-B_{1}$, $y_{2}\in B_{2}-B_{1}$, $(B_{1}-\{x\})\cup\{y_{1}\}\in\mathcal{B}(M)$ and $(B_{1}-\{x\})\cup\{y_{2}\}\in\mathcal{B}(M)$. Then the $y$ which satisfies $y\in B_{2}-B_{1}$ and $(B_{1}-\{x\})\cup\{y\}\in\mathcal{B}(M)$ is not unique.
\end{example}

If we require that such exchange of elements is unique, we obtain a special type of matroids.

\begin{definition}(Unique exchange matroid)
\label{definition117}
Let $M$ be a matroid. For any $B_{1}\in\mathcal{B}(M)$ and any $B_{2}\in\mathcal{B}(M)$, if by $x\in B_{1}-B_{2}$, $y_{1}\in B_{2}-B_{1}$, $y_{2}\in B_{2}-B_{1}$, $(B_{1}-\{x\})\cup\{y_{1}\}\in\mathcal{B}(M)$ and $(B_{1}-\{x\})\cup\{y_{2}\}\in\mathcal{B}(M)$, we obtain $y_{1}=y_{2}$, $M$ is called a unique exchange matroid.
\end{definition}

There is a simple relationship between unique expansion matroids and unique exchange matroids.

\begin{proposition}
\label{proposition118}
If $M$ is a unique expansion matroid, $M$ is a unique exchange matroid.
\end{proposition}

\begin{proof}

For any $B_{1}\in\mathcal{B}(M)$ and any $B_{2}\in\mathcal{B}(M)$, by $x\in B_{1}-B_{2}$, $y_{1}\in B_{2}-B_{1}$ and $y_{2}\in B_{2}-B_{1}$, we have that $B_{1}-\{x\}\in s(M)$, $y_{1}\in B_{2}$ and $y_{2}\in B_{2}$. Then by $(B_{1}-\{x\})\cup\{y_{1}\}\in\mathcal{B}(M)$ and $(B_{1}-\{x\})\cup\{y_{2}\}\in\mathcal{B}(M)$, we obtain that $y_{1}=y_{2}$. Thus $M$ is a unique exchange matroid. $\Box$

\end{proof}

The converse of the above proposition is not true. Let us see an example.

\begin{example}
\label{example119}
Let $E=\{1,2,3,4\}$, $\mathcal{B}=\{\{1,2,3\},\{1,2,4\},\{1,3,4\}\}$ and $M=(E,Low\\
(\mathcal{B}))$. Then $M\in B_{u}$. But by $3\in\{1,3,4\}\in\mathcal{B}(M)$, $4\in\{1,3,4\}\in\mathcal{B}(M)$, $\{1,2\}\in s(M)$, $\{1,2,3\}\in\mathcal{B}(M)$ and $\{1,2,4\}\in\mathcal{B}$, we know that $M\notin S_{b}$.
\end{example}

The following proposition gives a relationship between the dual matroids of unique expansion matroids and unique exchange matroids.

\begin{theorem}
\label{theorem120}
If $M$ is a unique expansion matroid, $M^{*}$ is a unique exchange matroid.
\end{theorem}

\begin{proof}
By Theorem~\ref{theorem52}, we know that there exists some partition $P$ which satisfies $\cup P\subseteq E(M)$ and $M=M_{E(M)}(P)$. Let $B_{1},B_{2}\in\mathcal{B}(M^{*})$ and $x\in B_{1}-B_{2}$. Let $y_{1}\in B_{2}-B_{1}$ and $(B_{1}-\{x\})\cup\{y_{1}\}\in\mathcal{B}(M^{*})$. It is obvious $y_{1}\neq x$. By Proposition~\ref{proposition339}, we know that $E(M)-\cup P\subseteq B_{2}$. We claim that $x\in\cup P$. Otherwise, $x\in B_{2}$. It is contradictory. Then there exists one and only one block of $P$, say $D$, such that $x\in D$. By Proposition~\ref{proposition339}, we know that $|D-B_{1}|=1$. We claim that $y_{1}\in D$. Otherwise, suppose $y_{1}\notin D$. Then $|D-((B_{1}-\{x\})\cup\{y_{1}\})|=|D-B_{1}|+1=2$. It is contradictory. By $y_{1}\in D-B_{1}$ and $|D-B_{1}|=1$, we know that $\{y_{1}\}=D-B_{1}$. Similarly, if $y_{2}\in B_{2}-B_{1}$ and $(B_{1}-\{x\})\cup\{y_{2}\}\in\mathcal{B}(M^{*})$, $\{y_{2}\}=D-B_{1}$. Therefore $y_{1}=y_{2}$. Thus $M^{*}$ is a unique exchange matroid. $\Box$

\end{proof}

The following example indicates that the converse of the above proposition is not true.

\begin{example}
\label{example121}
Let $E=\{1,2,3,4\}$, $\mathcal{B}=\{\{1,2\},\{1,3\},\{1,4\}\}$ and $M=(E,Low
(\mathcal{B}))$. Then $M$ is a unique exchange matroid, but $M^{*}$ is not a unique expansion matroid.
\end{example}

\section{Conclusions }
\label{S:Conclusions}

From the viewpoint of the expansion uniqueness of the independent sets of matroids, this paper defined unique expansion matroids. From the viewpoint of the minimality of the base families of matroids, this paper defined union minimal matroids. Some properties of these two types of matroids were given. We proved that unique expansion matroids are union minimal matroids. In addition, we extend the concept of unique expansion matroid to unique exchange matroid. All of these open up a new point for understanding matroids.

\section*{Acknowledgments}
This work is supported in part by the National Natural Science Foundation of China under Grant No. 61170128, the Natural Science Foundation of Fujian Province, China, under Grant No. 2012J01294, and the Science and Technology Key Project of Fujian Province, China, under Grant No. 2012H0043.


\begin{thebibliography}{10}
\expandafter\ifx\csname url\endcsname\relax
  \def\url#1{\texttt{#1}}\fi
\expandafter\ifx\csname urlprefix\endcsname\relax\def\urlprefix{URL }\fi

\bibitem{Edmonds71Matroids}
J.~Edmonds, Matroids and the greedy algorithm, Mathematical Programming 1~(1)
  (1971) 127--136.

\bibitem{HuangZhu12Geometriclattice}
A.~Huang, W.~Zhu, Geometric lattice structure of covering-based rough sets
  through matroids, Journal of Applied Mathematics 2012 (2012) Article ID
  236307, 25 pages.

\bibitem{Lai01Matroid}
H.~Lai, Matroid theory, Higher Education Press, Beijing, 2001.

\bibitem{Lawler01Combinatorialoptimization}
E.~Lawler, Combinatorial optimization: networks and matroids, Dover
  Publications, 2001.

\bibitem{LiuChen94Matroid}
G.~Liu, Q.~Chen, Matroid, National University of Defence Technology Press,
  Changsha, 1994.

\bibitem{LiuZhuZhang12Relationshipbetween}
Y.~Liu, W.~Zhu, Y.~Zhang, Relationship between partition matroid and rough set
  through k-rank matroid, Journal of Information and Computational Science 8
  (2012) 2151--2163.

\bibitem{TangSheZhu12matroidal}
J.~Tang, K.~She, W.~Zhu, Matroidal structure of rough sets from the viewpoint
  of graph theory, Journal of Applied Mathematics 2012 (2012) Article ID
  973920, 27 pages.

\bibitem{WangZhu11Matroidal}
S.~Wang, W.~Zhu, Matroidal structure of covering-based rough sets through the
  upper approximation number, International Journal of Granular Computing,
  Rough Sets and Intelligent Systems 2~(2) (2011) 141--148.

\bibitem{Whitney1935Ontheabstract}
H.~Whitney, On the abstract properties of linear dependence, American Journal
  of Mathematics 57 (1935) 509--533.

\bibitem{ZhuWang11Rough}
W.~Zhu, S.~Wang, Rough matroid, in: Granular Computing, (2011) 817--822.

\end{thebibliography}

\end{document}